\documentclass{article}
\usepackage{graphicx} 

\usepackage[english]{babel}
\usepackage{enumerate}
\usepackage{amsmath, amssymb, amsthm,  algpseudocode, algorithm}

\usepackage[a4paper,top=3cm,bottom=3cm,left=3cm,right=3cm,marginparwidth=1.75cm]{geometry}
\usepackage[textsize=small]{todonotes}
\usepackage[round]{natbib}
\usepackage{pifont}
\usepackage{bm}

\renewcommand{\qedsymbol}{$\blacksquare$}
\usepackage{amsfonts}
\usepackage{bbm}
\usepackage{enumitem}
\usepackage{graphicx}
\usepackage[colorlinks=true, allcolors=blue]{hyperref}
\usepackage{mathtools}
\usepackage{breqn}
\usepackage{float}
\usepackage{xcolor}
\usepackage{tcolorbox}

\newtheorem{theorem}{Theorem}[section]

\newtheorem{proposition}[theorem]{Proposition}

\newtheorem{remark}[theorem]{Remark}

\theoremstyle{definition}

\numberwithin{equation}{section}

\newcommand{\upperRomannumeral}[1]{\uppercase\expandafter{\romannumeral#1}}
\newcommand{\lowerRomannumeral}[1]{\lowercase\expandafter{\romannumeral#1}}

\newcommand{\R}{{\mathbb R}}

\usepackage{mathtools}
\usepackage{authblk}
\mathtoolsset{showonlyrefs}

\title{Simulation of square-root processes  made simple: applications to the Heston model}
\author{Eduardo Abi Jaber\thanks{eduardo.abi-jaber@polytechnique.edu. I am grateful for the financial support from the Chaires FiME-FDD, Financial Risks, Deep Finance \& Statistics and Machine Learning and systematic methods in finance at Ecole Polytechnique. I would like to thank Aurélien Alfonsi, Elie Attal,  Nathan De Carvalho, Louis-Amand Gérard, Shaun Li and Dimitri Sotnikov for fruitful discussions, as well as two anonymous referees for their valuable feedback.}}
\affil{Ecole Polytechnique, CMAP}

\begin{document}

\maketitle

\begin{abstract}
We introduce a  simple, efficient and accurate 
 numerical scheme that preserves non-negativity for simulating the square-root  process.   The novel idea is to simulate the integrated square-root  process first instead of the square-root  process itself.  Numerical experiments on realistic parameter sets, applied for the integrated process and the Heston model,  display high precision with a very low number of time steps. As a bonus, our scheme yields the exact limiting Inverse Gaussian distributions of the integrated square-root process with only one single time-step in two scenarios: (i) for high mean-reversion and volatility-of-volatility regimes, regardless of maturity; and (ii) for long maturities, independent of the other parameters.
\end{abstract}

\section*{Introduction}
Square-root processes are a cornerstone of stochastic modeling in finance, playing a central role in applications such as  interest rates \citep*{cox1985theory},  credit risk \citep*{duffie2012credit} and volatility \citep*{heston1993closed}. Despite their foundational importance, their simulation is often perceived as notoriously challenging. This work demonstrates that this perception is unwarranted by presenting an extremely  simple, accurate and nonnegative preserving  discretization scheme. The core idea is to simulate the integrated process first instead of the process itself. 

The square-root process $V$ is of the form:
\begin{align}\label{eq:cir}
V_t = V_0 + \int_0^t (a + b V_s) \, ds + c \int_0^t \sqrt{V_s} \, dW_s,
\end{align}
where $V_0, a \geq 0$, $b, c \in \mathbb{R}$, and $W$ is a standard Brownian motion.

The process exhibits an affine structure, meaning that both the instantaneous drift and the squared instantaneous diffusion are affine functions of $V$. This structure presents a double-edged sword:

\begin{itemize}
    \item \textbf{Tractability for pricing:} The affine structure leads to a  semi-analytical solution for the characteristic function, and opens the door to  efficient Fourier-based option pricing methods, see \cite*{cox1985theory,duffie2003affine, heston1993closed} and more generally \cite*{duffie2003affine}.
    \item \textbf{Simulation challenges:} The same affine structure introduces a square-root term in the diffusion coefficient, posing numerical challenges, especially in maintaining the non-negativity of a discretized version of $V$ of Euler-Maruyama type.
\end{itemize}

 \textbf{Literature review.} There is a substantial body of literature on the simulation of square-root processes.  To describe them, we define  the time-integrated process $U$ and the Brownian-integrated process  $Z$ by
\begin{align}\label{eq:UZ}
U_{s,t} := \int_s^t V_r \, dr, \quad Z_{s,t} := \int_s^t \sqrt{V_r} \, dW_r, \quad s \leq t.    
\end{align}  Despite differences in the proposed schemes, the approaches generally follow these steps:
\begin{itemize}
    \item Step 1: Update $V_{t_{i+1}}$ given $V_{t_i}$.
    \item Step 2: Sample $U_{t_i,t_{i+1}}$ knowing $(V_{t_i}, V_{t_{i+1}})$.
    \item Step 3: Deduce $Z_{t_i,t_{i+1}}$ using \eqref{eq:cir}, by setting:
    $$
    Z_{t_i,t_{i+1}} = \frac{1}{c} \left(V_{t_{i+1}} - V_{t_i} - a (t_{i+1} - t_i) - b U_{t_i,t_{i+1}} \right).
    $$
\end{itemize}

The schemes in the literature can be broadly categorized as follows:

\begin{description}
    \item[Category 1 - dynamics-based:] Methods that update $V_{t_{i+1}}$ knowing $V_{t_i}$ in Step 1 using (improved) Euler-Maruyama-type discretizations  of the stochastic differential equation of $V$ in \eqref{eq:cir}. And then approximate the integral $U_{t_i,t_{i+1}}$ from $(V_{t_i}, V_{t_{i+1}})$ in Step 2  using for instance left-point or mid-point rule. The advantage of such methods is usually their relative simplicity and interpretability. This category includes the works of \cite*{alfonsi2005discretization, alfonsi2010high, berkaoui2008euler, deelstra1998convergence, gyongy2011note, higham2005convergence, kahl2006fast, neuenkirch2014first} among many others. A limitation of this category is that not all Euler-type methods preserve the non-negativity of $V$ and can sometimes introduce important biases.   
    \item[Category 2 - distribution-based:] Methods that sample $V_{t_{i+1}}$  knowing $V_{t_i}$ in Step 1 using its exact (or approximate) distribution, which is a noncentral chi-squared distribution, which is relatively  efficient numerically.   And then sample $U_{t_i,t_{i+1}}$  based on either the exact or (approximate) distribution of $U_{t_i,t_{i+1}}$ conditional  on the endpoints  $(V_{t_i}, V_{t_{i+1}})$. Simulating from the exact distribution is possible as shown by  \cite{broadie2006exact}, however it is computationally expensive.  Later works initiated by \cite{glasserman2011gamma} introduced approximations to improve efficiency. Techniques such as matching moments coupled with educated guesses of the   conditional distribution of $U_{t_i,t_{i+1}}$  on   $(V_{t_i}, V_{t_{i+1}})$ have been proposed: with an  Inverse Gaussian distribution  by \cite{tse2013low}, Inverse Gamma distribution by \cite*{begin2015simulating}, Poisson conditioning of \cite{choi2024simulation} distributions, discrete random variables by \cite{lileika2020weak}. These methods offer a good trade-off between computational time and accuracy, but they usually require some amount of offline/online pre-computations and lack transparency in their mathematical link with the dynamics in \eqref{eq:cir}. 
\end{description}
   In some  cases, both categories are combined in different steps. The most notable example is the Quadratic Exponential (QE) scheme  developed by \cite{andersen2007efficient}  which samples $V_{t_{i+1}}$ knowing $V_{t_i}$ using a switching between a squared Gaussian (when $V_{t_i}$ is large enough) or a  tweaked exponential distribution   (when  $V_{t_i}$ is small)  with moment matching for Step 1; and then  uses the mid-point rule to approximate $U_{t_i,t_{i+1}}$ by  $(V_{t_i} + V_{t_{i+1}})/2$ in Step 2. Despite its hybrid nature, the methodology still adheres to the conventional structure of Steps 1, 2, and 3. The QE scheme seems quite  popular among practitioners and plays the role of a benchmark in the academic literature.  In the same vein, \cite{zhu2011simple} approximates the volatility process, rather than the variance process, using an Ornstein-Uhlenbeck process with moment matching. \cite{van2010efficient} proposes various hybrid approaches to accelerate the discretizations of Steps 1 and 2. Finally, we also mention the comparison study in \cite{possamai2011efficient}.

While these methodologies are typically treated as distinct, there has been limited investigation into the connections between the dynamics-based and distributional approaches. In a perfect world, one would like to read-off the distributional properties from Euler-type discretization method. This work aims to bridge that gap by establishing a clear link between discretized dynamics and distributional properties, thereby offering a deeper understanding of the underlying affine dynamics.

This work aims to reconcile the affine structure with Euler-Maruyama simulation schemes. Specifically, we ask:
\begin{center}
    \textit{Is there an unexplored way to exploit the affine structure for discretizing the process in an
Euler-Maruyama fashion?}\\
 \textit{Can we construct a simple, efficient, and accurate scheme for simulating square-root processes?}
\end{center}

We answer these questions \textit{affirmatively}.

\textbf{Main contributions.} 
Our approach builds upon simulating the integrated quantities \eqref{eq:UZ} first before the process $V$.  These quantities not only capture the essential information required for many financial applications but are also critical for simulating other processes, such as the \cite{heston1993closed} model.

We already have all the ingredients needed to present the discretization scheme we propose, which emphasizes the simplicity of the approach.  The following algorithm recursively constructs $ (\widehat{V}_i)_{i=0, \ldots, n} $, $ (\widehat{U}_{i, i+1})_{i=0, \ldots, n-1} $, and $ (\widehat{Z}_{i, i+1})_{i=0, \ldots, n-1} $. We coin it the \textbf{iVi} scheme for \textbf{i}ntegrated \textbf{$\bold V$} (or \textbf{V}ariance) \textbf{i}mplicit scheme.

\begin{algorithm}[H]
\caption{ \textbf{- The iVi scheme:} Simulation of \( \widehat V, \widehat U, \widehat Z \)}\label{alg:simulation}
\begin{algorithmic}[1]
\State \textbf{Input:} Initial value \( \widehat V_0 = V_0 \), partition \( 0 = t_0 < t_1 < \cdots < t_n = T \), parameters \( a, b, c \).
\State \textbf{Output:} \( \widehat V_{i+1}, \widehat U_{i, i+1}, \widehat Z_{i, i+1} \) for \( i = 0, \ldots, n-1 \).

\For{$i = 0$ to $n-1$}
    \State Compute the quantities:
\begin{align}\label{eq:alphasigma}
    \alpha_i ={\widehat V_i} \frac{e^{b(t_{i+1}-t_i) }- 1}b  + \frac{a}{b} \left( \frac{e^{b(t_{i+1}-t_i)} - 1}b  - (t_{i+1}-t_i) \right), \quad \sigma_i = c \frac{e^{b(t_{i+1} - t_i)} - 1}{b}.    
    \end{align}

    \State Simulate the increment of the integrated process:
    \begin{align}\label{eq:hatUsample}
   \widehat U_{i, i+1} \sim IG\left(\alpha_i, \left(\frac{\alpha_i}{\sigma_i}\right)^2 \right).
     \end{align}
    \State Set:
    \begin{align}\label{eq:Zii}
    \widehat Z_{i, i+1} = \frac{1}{\sigma_i} \left(\widehat U_{i, i+1} - \alpha_i\right)    
    \end{align}
    \State Update  the instantaneous process:
  \begin{align}\label{eq:Vii}
   \widehat V_{i+1} =\widehat V_i + a (t_{i+1} - t_i) + b \widehat U_{i, i+1} + c \widehat Z_{i, i+1}.
 \end{align}
\EndFor
\end{algorithmic}
\end{algorithm}
 Here, $IG$ refers to an Inverse Gaussian distribution, whose definition and a simple sampling algorithm are provided in Appendix~\ref{A:IG} and with the convention in \eqref{eq:alphasigma} that for $b=0$:  $\frac{e^{bt}-1}b=t$ and $\frac{ \frac{e^{bt} - 1}b  - t }{b}  = \frac{t^2}2$.

Aside from the inherent simplicity of Algorithm~\ref{alg:simulation}, which is both efficient and free from precomputations and fine-tuning of hyperparameters, we show that the  iVi scheme enjoys the following key features:
\begin{itemize}
    \item The scheme is built from straightforward right-endpoint Euler-type discretization rule applied to a single integral in the dynamics of $U$ and leverages the affine structure on the level of the dynamics of $(U,Z)$. 
    \item The scheme ensures the  non-negativity of the process $V$, i.e., $\widehat{V}_i \geq 0$ for all $i = 0, \ldots, n$, as established in Theorem~\ref{T:nonnegative}.
    \item The scheme captures essential distributional properties:
    \begin{itemize}
        \item[(i)] The first conditional moments  are perfectly captured, as proved in Proposition~\ref{P:moments}.
        \item[(ii)] The Inverse Gaussian distribution emerges naturally from the right-endpoint discretization of the conditional characteristic function of $U$, as shown in Remark~\ref{R:charcomparison}.
        \item[(iii)] The scheme accurately reproduces the Inverse Gaussian limiting distribution of $U_{0,T}$ in market regimes characterized by large mean reversion and high volatility of volatility, in line with the first observations made by \cite{mechkov2015fast} and the works of     \cite{abijaber2024reconciling, mccrickerd2019foundations}, see Remark~\ref{R:largemeanrev}.
        \item[(iv)] It also provides the exact Inverse Gaussian limiting distribution of $U_{0,T}$ for large maturities, as established by  \cite*{forde2011large}, Remark~\ref{R:largemeanrev}.
    \end{itemize}
    \item In terms of performance, the iVi scheme displays high precision with very few time-steps,  for the integrated process and the Heston model,  even under realistic and challenging scenarios, see Sections~\ref{S:U} and \ref{S:HestonNumerics}. As illustration, Figure~\ref{fig:intro} provides  six slices of the volatility surface in the Heston model (black) with challenging parameters calibrated to the market. Our iVi scheme (in orange) is computed with only one single time step  (equal to the maturity $T$) per slice! 
\end{itemize}

 \begin{figure}[h!]
 	\centering
 	\includegraphics[width=.8\textwidth]{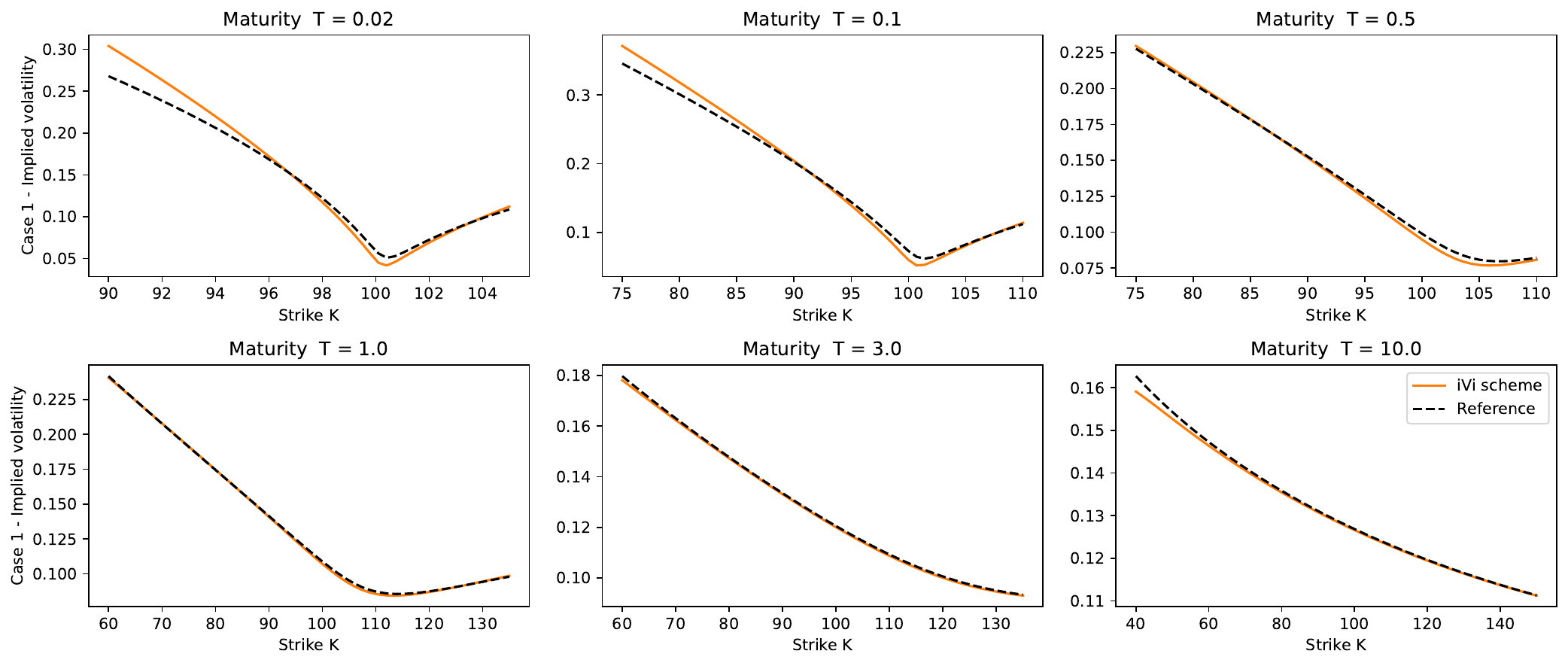} 
 	\caption{Heston's Implied volatility surface with the iVi scheme with one single time step and 2 million sample paths for model parameters as in Case 1 in Table~\ref{tab:parameter_cases}.}
 	\label{fig:intro} 
 \end{figure}

\textbf{Our approach.} Unlike the conventional methodology, our approach departs from the usual Step 1, 2, and 3 order. Instead of first sampling $V_{t_{i+1}}$, we focus on sampling $U_{t_i,t_{i+1}}$ given $V_{t_i}$ by discretizing the dynamics of $U$ using a simple Euler-Maruyama right-endpoint scheme. Leveraging the affine structure, we derive an intertwined relationship between $U_{t_i,t_{i+1}}$ and $Z_{t_i,t_{i+1}}$ in the form of the implicit equation:
\[
U_{t_i,t_{i+1}} \approx \alpha_i + \sigma_i Z_{t_i,t_{i+1}}, \quad \text{and} \quad [Z]_{t_i,t_{i+1}} = U_{t_i,t_{i+1}},
\]
where $[\cdot]$ denotes the quadratic variation. 

A solution (among many) to this equation is given by an Inverse Gaussian distribution, leading to the sampling formula in \eqref{eq:hatUsample}. Using this relationship, the definition of $Z_{t_i,t_{i+1}}$ in \eqref{eq:Zii} becomes clear, and the update of $V_{t_{i+1}}$ in \eqref{eq:Vii} follows naturally from the dynamics \eqref{eq:cir}. To justify the choice of the Inverse Gaussian distribution, we exploit the expression of the characteristic function of $U_{t_i,t_{i+1}}$ given $V_{t_i}$ in terms of Riccati equations. We demonstrate that applying the same right-endpoint rule at the level of the Riccati equations naturally leads to an Inverse Gaussian distribution. 

This novel perspective not only simplifies the simulation process but also provides deeper insights, which might be useful in other contexts, into the interplay between the dynamics and distributional properties of the square-root process—or more precisely, of the integrated square-root process!

\textbf{Outline.} In Section~\ref{S:scheme} we provide the mathematical derivation of the scheme as well as its distributional properties. Numerical illustrations for the scheme for the integrated process $U$ and the Heston model are provided in Sections~\ref{S:U} and \ref{S:HestonNumerics}. Useful properties of Inverse Gaussian distribution and the  Heston model are collected in the appendices, together with calibrated volatility surfaces for the parameters used in our experiments.

\section{The iVi scheme}\label{S:scheme}
 In this section, we will detail the mathematical derivation of the iVi scheme, prove that it  preserves non-negativity for the discretized process $\widehat V$ and study some of its distributional properties. 

\subsection{Deriving the iVi scheme: one integral to approximate}
The first step it to write the dynamics of the integrated process $U_{s,t} = \int_s^t V_r dr$  between $s$ and $t$. This is done by first writing the variation of constants formula for the process $V$ in \eqref{eq:cir} as 
\begin{align}\label{eq:cirvariation}
   V_r = V_s e^{b(r-s)} + a  \frac{e^{b(r-s)} - 1} b 
   + c \int_s^r e^{b(r-u)} \, dZ_{s,u}, \quad s \leq r,
\end{align}
where $dZ_{s,u}$ is the differential with respect to the second variable $u$, for $u\geq s$, i.e.~$dZ_{s,u}= \sqrt{V_u}dW_u$. Integrating $V_r$ between $s$ and $t$ and applying stochastic Fubini's theorem yields the following dynamics for $U$.
\begin{proposition}  The dynamics of the process $U$ defined in \eqref{eq:UZ} are given by 
\begin{align}\label{eq:Ust}
    U_{s,t} = {V_s} \frac{e^{b(t-s)} -1 }b  + \frac{a}{b} \left( \frac{e^{b(t-s)} - 1}b  - (t-s) \right)  + c \int_s^t e^{b(t-r)}  Z_{{s,r}}  \,dr, \quad s\leq t,
\end{align}
where we recall  the convention that for $b=0$:  $\frac{e^{bt}-1}b=t$ and $\frac{ \frac{e^{bt} - 1}b  - t }{b}  = \frac{t^2}2$.
\end{proposition}

\begin{proof}
Fix $s\leq t$.
      Integrating the dynamics of $V$ in \eqref{eq:cirvariation} between $s$ and $t$ gives
   $$
   U_{s,t} = \int_s^t V_r \, dr =   {V_s} \frac{e^{b(t-s)} - 1 }b  + \frac{a}{b} \left( \frac{e^{b(t-s)} - 1}b  - (t-s) \right) 
   + c \int_s^t \int_s^r e^{b(r-u)} \, dZ_{s,u}\, dr.
   $$
Successive applications of stochastic Fubini's theorem and a change of variables (using that $Z_{s,s}=0$), on the last term yield 
\begin{align*}
    \int_s^t \int_s^r e^{b(r-u)} \,dZ_{s,u}\, dr &= \int_s^t \int_u^t e^{b(r-u)} \,dr  \,dZ_{s,u}\\  
    &= \int_s^t \int_0^{t-u} e^{bv} \,dv  \,dZ_{s,u}\\
    &=  \int_0^{t-s} \int_s^{t-v}  \,dZ_{s,u} e^{bv} \,dv\\ &= \int_0^{t-s} Z_{s,{t-v}} e^{bv} \,dv\\
    &= \int_s^t Z_{s,u}e^{b(t-u)}du,
\end{align*}
which ends the proof. 
\end{proof}

The main idea behind our scheme is to simply discretize the equation \eqref{eq:Ust} for the integrated process $U$ between $t_{i}$ and $t_{i+1}$, assuming the knowledge of $\hat V_i \approx V_{t_i}$.  There is only one integral to discretize. Using the right endpoint rule on $Z_{t_i,\cdot}$, i.e.~approximating $Z_{t_i,s} $ by the value $Z_{t_i,t_{i+1}}$ for $s\in [t_i,t_{i+1})$, yields
$$
\int_{t_i}^{t_{i+1}} e^{b(t_{i+1}-r)}  Z_{t_i,r} \, dr \approx \int_{t_i}^{t_{i+1}} e^{b(t_{i+1}-r)}   \, dr Z_{{t_i,t_{i+1}}} = \frac{e^{b(t_{i+1}-t_i)} -1}b  Z_{{t_i,t_{i+1}}}.
$$
Plugging this in \eqref{eq:Ust}, yields 
\begin{align}\label{eq:Uti}
     U_{{t_i},{t_{i+1}}}  \approx  \hat V_i \frac{e^{b(t_{i+1}-t_i)}- 1}b  + \frac{a}{b} \left( \frac{e^{b(t_{i+1}-t_i)} - 1}b  - (t_{i+1}-t_i) \right) 
   + c \frac{e^{b(t_{i+1}-t_i)} -1}b   Z_{{t_i,t_{i+1}}}.
\end{align}
Now the key point is to observe that this forms an implicit equation on $ U_{{t_i},{t_{i+1}}}$ by rewriting $Z$ using a time-changed Brownian motion. Indeed, recall from the definition of $(Z_{t_i,s})_{s\geq t_i}$ in \eqref{eq:UZ} that it is a  continuous local martingale with quadratic variation   $\int_{t_i}^{\cdot} V_r dr = U_{t_{i},\cdot}$. Hence, an application of the celebrated Dambis, Dubins-Schwarz Theorem see \cite[Theorem 1.1.6]{revuz2013continuous} yields the representation of $(Z_{t_i,s})_{s\geq t_i}$ in terms of a time changed Brownian motion $Z_{t_i, s} = \widetilde W_{U_{t_i,s}}$ for all $s\geq t_{i}$ where $\widetilde W$ is a standard Brownian motion. 

Hence, plugging this in \eqref{eq:Uti},  approximating $U_{t_i, t_{i+1}}$  boils down to finding  a nonnegative random variable $\widehat U_{i,i+1}$ solving  
\begin{align}\label{eq:IGequation}
    \widehat U_{i,i+1} = \alpha_i + \sigma_i \widetilde W_{\widehat U_{i,i+1}},
\end{align}
with $\alpha_i$ and $\sigma_i$ as in \eqref{eq:alphasigma}.  

Said differently $\widehat U_{i,i+1}$ is a passage  time of the level $\alpha_i$ for the drifted Brownian motion $ (s -  \sigma_i \widetilde W_{s})_{s\geq 0} $. In particular,  the first passage time 
$$ X = \inf \left\{ s\geq 0:  s -  \sigma_i \widetilde W_{s}= \alpha_i \right\} $$
satisfies \eqref{eq:IGequation}. 
It is well known  that $X$ follows an  Inverse Gaussian distribution $IG(\mu_i,\lambda_i)$ with mean parameter $\mu_i= \alpha_i$  and shape parameter $\lambda_i = \frac{\alpha_i^2}{\sigma_i^2}$. Inverse Gaussian distributions are recalled in Appendix~\ref{A:IG}.  

Hence, we sample $\widehat U_{i,i+1}$ using  the Inverse Gaussian distribution as in \eqref{eq:hatUsample}. Beyond its tractability and efficient  sampling, see Algorithm~\ref{alg:IG_sampling}, the choice of the Inverse Gaussian distribution is further justified in Section~\ref{S:whyIG} where it is shown to be intimately linked with the distributional properties of the integrated square-root process $U$. Then, using \eqref{eq:IGequation} we can define the   $\widehat Z_{i,i+1}$ (which plays the role of $\widetilde W_{\widehat U_{i,i+1}}$) using equation \eqref{eq:Zii}. Finally, to update $\widehat V$, we write the square-root equation \eqref{eq:cir} between $t_i$ and $t_{i+1}$ as
$$ V_{t_{i+1}} = V_{t_i} + a (t_{i+1}-t_i) + b U_{t_i, t_{i+1}} + c Z_{t_i, t_{i+1}}. $$
This yields the update formula for $\widehat V_{i,i+1}$ in \eqref{eq:Vii} follows.

Putting everything together, we arrive at the iVi scheme of Algorithm~\ref{alg:simulation}.  

\begin{remark}
A simpler scheme can be derived without using the variation of constants formula \eqref{eq:cir}, after a simple integration of the equation \eqref{eq:cir} for $V$ between $s$ and $t$, to get
\begin{align}\label{eq:Unovar}
	U_{s,t} = V_s (t-s) + a \frac{(t-s)^2}{2} + \int_s^t (bU_{s,r}  + cZ_{s,r}) \, dr.
\end{align}
	Using the right endpoint to approximate the integral between $t_i$ and $t_{i+1}$ in \eqref{eq:Unovar},  we arrive to the implicit scheme
	$$ U_{i,i+1} = V_i (t_{i+1}-t_i) + a \frac {(t_{i+1}-t_i) ^2} 2  + b  (t_{i+1}-t_i)    U_{i,i+1} + c (t_{i+1}-t_i) Z_{i,i+1}.$$
	Hence, as long as $(1 -  b(t_{i+1}-t_i)) >0$ (which is the case for mean reverting dynamics, i.e.~$b \leq 0$), using the same ideas as above we get to 
	 \begin{align}
		\widetilde U_{i, i+1} \sim IG\left(\frac{\widetilde \alpha_i}{(1 - b(t_{i+1} - t_i))}, \left(\frac{\widetilde \alpha_i}{\widetilde \sigma_i}\right)^2 \right),
	\end{align}
with  
$$  \widetilde \alpha_i = \widetilde V_i (t_{i+1}-t_i) + a \frac {(t_{i+1}-t_i) ^2} 2, \quad \widetilde \sigma_i = c (t_{i+1}-t_i),$$
and 
	\begin{align}
		\widetilde Z_{i, i+1} = \frac{1}{\widetilde 
 \sigma_i} \left( (1 - b(t_{i+1} -t_i)) \widetilde U_{i, i+1} - \widetilde \alpha_i\right).    
	\end{align}
	Then, $\widetilde V_{i+1}$ is updated using \eqref{eq:Vii} but with $(\widetilde V,\widetilde U,\widetilde Z)$ instead of $(\widehat V,\widehat U,\widehat Z)$. It should be clear that such scheme introduces an additional bias compared to Algorithm~\ref{alg:simulation} since,  the drift here is approximated instead of being exactly solved for, using the variation of constants formula. \qedsymbol 
\end{remark}

\subsection{Well-definedness and non-negativity of $\widehat V$}

To ensure the well-definedness of Algorithm~\ref{alg:simulation}, we still have to check that the mean parameter $\alpha_i$ in \eqref{eq:alphasigma} of the Inverse Gaussian distribution is nonnegative for each $i=0,\ldots, n$, with the convention that $IG(0,0)$ is equal  to $0$. 
This is the object of the next theorem, and as a by-product, we will obtain that the discretized square-root process $\widehat V$ satisfies $\widehat V_i \geq 0$ for all $i=0,\ldots, n$.

\begin{theorem}\label{T:nonnegative}
    Let $V_0,a \geq 0$ and $b,c \in \mathbb R$.  Consider  \(  (\widehat V_i)_{i=0, \ldots, n} \), \( ( \widehat U_{i, i+1})_{i=0, \ldots, n-1} \), and \( ( \widehat Z_{i, i+1})_{i=0, \ldots, n-1} \) satisfying the recursions of the {\normalfont iVi} scheme in  Algorithm~\ref{alg:simulation}. Then,  we have that 
    \begin{align}
       \widehat V_i,  \alpha_i \geq 0, \quad i=0,\ldots, n.
    \end{align}
\end{theorem}

\begin{proof}
Since $a\geq 0$, it follows from the definition of $\alpha_i$ in \eqref{eq:alphasigma} that $\alpha_i\geq 0$ if $\widehat V_i \geq 0$, for all $i=0,\ldots, n$. Hence, it suffices to prove  that 
$$  \widehat V_i \geq 0, \quad i=0,\ldots, n. $$
First we note that for $i=0$, we have that 
$\widehat V_0 = V_0\geq 0$ by definition of $V_0$. Fix $i=0,\ldots, n-1$, let us show that $\widehat V_{i+1} \geq 0$ using \eqref{eq:Vii}. Plugging   in \eqref{eq:Vii}, the  expression for  $Z_{i,i+1}$ given by \eqref{eq:Zii}  and that of $\alpha_{i} \in \eqref{eq:alphasigma}$,  we  arrive to  
$$ V_{i+1} =  \frac{b e^{b(t_{i+1}- t_i)}}{e^{b(t_{i+1}- t_i)} - 1} U_{i,i+1} + \frac a b \left[\frac{b(t_{i+1}-t_i)  e^{b(t_{i+1}- t_i)}}{e^{b(t_{i+1}- t_i)} - 1} - 1 \right] $$
the first term is nonnegative for all $b \in \R$, since $U_{i,i+1} \geq 0$ by definition of an Inverse Gaussian distribution  and the second term is also nonnegative by using that $a\geq 0$ and observing that the function $f(x):= \frac{xe^{x}}{e^x -1}$ satisfies $f(x) \leq 1$ for $x\leq 0$ and $f(x)\geq 1$ for $x\geq 0$ so that  
$\frac{1}{b}(f(b ({t_{i+1}-t_i})) - 1 )\geq 0$ independently of the sign of $b$.  This ends the proof. 
\end{proof}

Figure~\ref{fig:samplepaths} displays   sample paths of  $(\widehat V, \widehat U, \widehat Z)$ constructed using the iVi scheme in Algorithm~\ref{alg:simulation} for $a>0$ (reflecting boundary) and $a=0$ (absorbing boundary). Notice how in both cases, all sample paths of $\widehat V$ remain nicely nonnegative and $\widehat U$ non-decreasing. For the first row, $\widehat V$ bounces back when it reaches $0$ while it gets absorbed at $0$ for the second row.

\begin{figure}[h!]
    \centering
    \includegraphics[width=1.\textwidth]{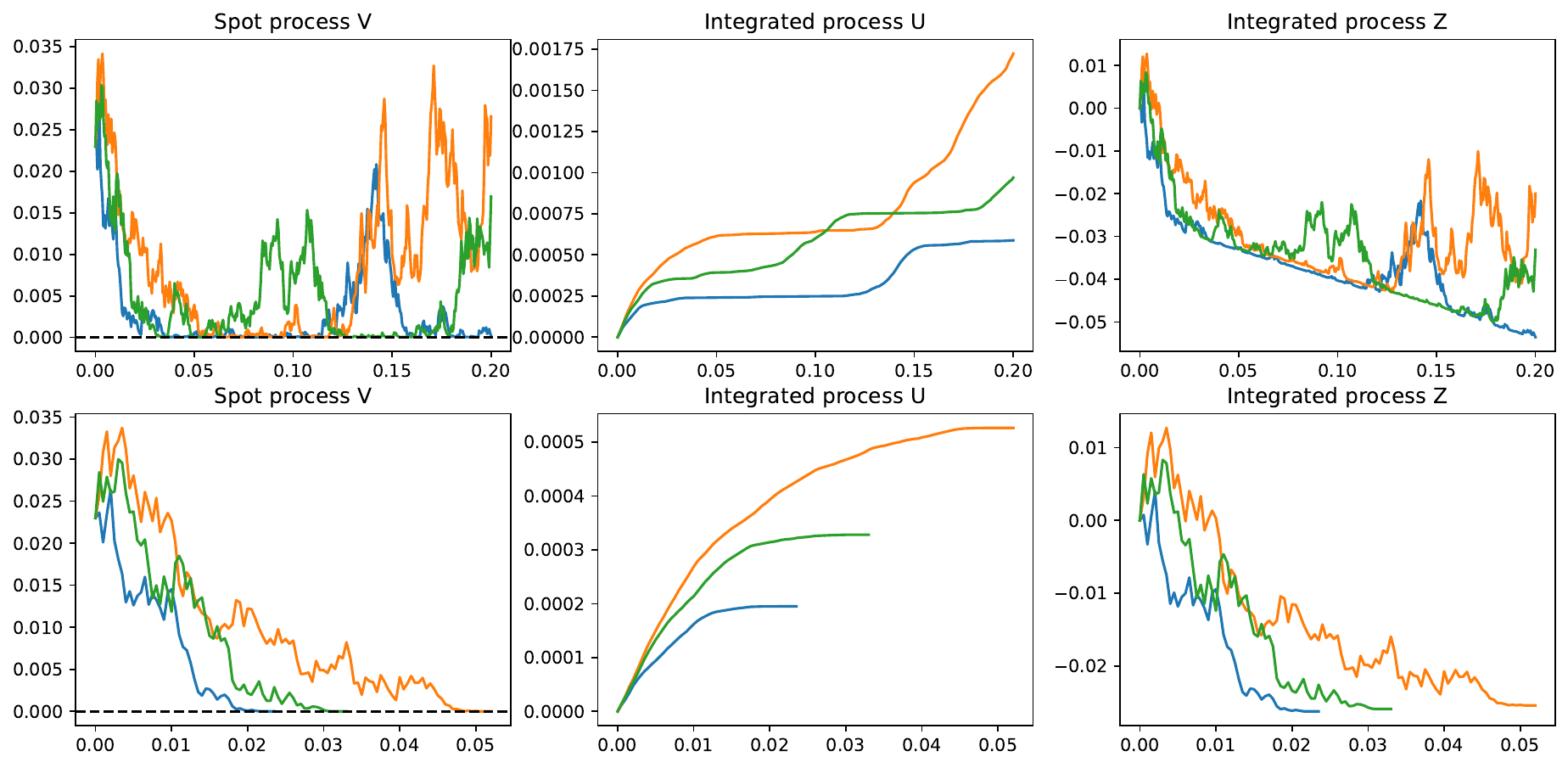} 
    \caption{Row 1 - Reflecting boundary $a>0$: Parameters as in Case 2 in Table~\ref{tab:parameter_cases} below. Row 2 - absorbing boundary: $a=0$ and other parameters unchanged. $T = 0.2$ and a $400$  time steps. }
    \label{fig:samplepaths} 
\end{figure}

\subsection{Distributional properties of the iVi scheme}

The Inverse Gaussian distribution is not the only one satisfying \eqref{eq:IGequation}, which naturally raises the question: 
\emph{Why choose the Inverse Gaussian distribution?} In this section, we provide three justifications based on distributional properties of the square-root process:  
\begin{itemize}
\item \textbf{First conditional moments} are perfectly captured for 
 the processes $(V,U,Z)$, as shown in  Proposition~\ref{P:moments}.
    \item \textbf{Discretization of the conditional characteristic function}  of $U$ naturally leads to an  Inverse Gaussian distribution, see Proposition~\ref{P:charcomparison} and Remark~\ref{R:charcomparison}.  
    \item \textbf{Convergence of the process $U$ towards an Inverse Gaussian process} happens in two configurations: (i) either in  certain regimes with large mean reversion and volatility of volatility for any maturity (ii) or for large maturities independently of the other parameters, refer to Remark~\ref{R:largemeanrev}.  
\end{itemize}

\subsubsection{First conditional moments are matched}

The next proposition shows that  the first conditional moments  of $(V,U,Z)$ are perfectly matched by those of $(\widehat V, \widehat U, \widehat Z)$. However, because of the approximation of the integral in \eqref{eq:Uti}, higher moments are not expected to be perfectly matched. 

\begin{proposition}\label{P:moments}
    Fix $i=0,\ldots, n-1$ and $v\geq 0$. Then, 
    \begin{align*}
        \mathbb E\left[ \widehat U_{i,i+1} \mid\widehat V_i =v \right] &=  \mathbb E\left[U_{t_i,t_{i+1}} \mid V_i =v \right] =  {v} \frac{e^{b(t_{i+1}-t_i)} - 1 }b  + \frac{a}{b} \left( \frac{e^{b(t_{i+1}-t_i)} - 1}b  - (t_{i+1}-t_i) \right), \\
         \mathbb E\left[ \widehat Z_{i,i+1} \mid \widehat V_i =v \right] &=  \mathbb E\left[Z_{t_i,t_{i+1}} \mid V_i =v \right] = 0,\\
          \mathbb E\left[ \widehat V_{i+1} \mid \widehat V_i =v \right] &=  \mathbb E\left[V_{t_i} \mid V_i =v \right] =  v e^{b(t_{i+1}-t_i)} + \frac{a}{b}(e^{b(t_{i+1}-t_i)} - 1).
    \end{align*}
\end{proposition}

\begin{proof}
The explicit expressions for the conditional moments of $(V,U,Z)$ are straightforward to obtain from   \eqref{eq:cirvariation}, \eqref{eq:UZ} and \eqref{eq:Ust}.  Let us look at those of $(\widehat V, \widehat U, \widehat Z)$.
    Conditional on $\widehat V_i$, it follows from \eqref{eq:hatUsample} that $\widehat U_{i,i+1}$  follows an Inverse Gaussian distribution with mean  $\alpha_i$, see \eqref{eq:IGmean}. Using the expression of $\alpha_i$ in \eqref{eq:alphasigma} we get 
    $$ \mathbb E\left[ \widehat U_{i,i+1}\mid\widehat V_i =v \right]  = {v} \frac{e^{b(t_{i+1}-t_i)} - 1}b  + \frac{a}{b} \left( \frac{e^{b(t_{i+1}-t_i)} - 1}b  - (t_{i+1}-t_i) \right),$$
    which proves the  equality for the first moment of  $U$.  Using the construction of $\widehat Z_{i,i+1}$ in \eqref{eq:Zii} it follows that 
    $ \mathbb E\left[ \widehat Z_{i,i+1} \mid\widehat V_i =v \right] = 0$, which shows the  equality for the first moment of $Z$. Finally, using \eqref{eq:Vii}, we can compute 
    \begin{align*}
         \mathbb E\left[ \widehat V_{i+1} \big | \widehat V_i =v \right] &= v + a(t_{i+1} - t_i) + b  \mathbb E\left[ \widehat U_{i,i+1} \big | \widehat V_i =v \right] + c  \mathbb E\left[ \widehat Z_{i,i+1} \big | \widehat V_i =v \right]\\
         &=  v + a(t_{i+1} - t_i) + {v} (e^{b(t_{i+1}-t_i)} - 1)  + a \left( \frac{e^{b(t_{i+1}-t_i)} - 1}b  - (t_{i+1}-t_i) \right)\\
         &=  v e^{b(t_{i+1}-t_i)} + \frac{a}{b}(e^{b(t_{i+1}-t_i)} - 1),
    \end{align*}
    which proves the equality for the first moment of $V$ and ends the proof. 
\end{proof}

\subsubsection{The distributions are intimately linked}\label{S:whyIG}

To begin, the next proposition highlights striking similarities between the conditional characteristic functions of $\widehat{U}$ and $U$: both exhibit an exponentially affine dependence on $v$. More importantly, the Inverse Gaussian distribution emerges naturally as an implicit Euler-type discretization of the Riccati equations governing the  conditional characteristic function of $U$ between two consecutive time steps.   In what follows, we use the principal branch for the complex square-root.

\begin{proposition}\label{P:charcomparison}
   Fix $w\in \mathbb C$ such that $\Re(w)\leq 0$. Fix $i=0,\ldots, n-1$ and $v\geq 0$. Then, 
    \begin{align}
        \mathbb E\left[ \exp\left( w\widehat U_{i,i+1}\right) \mid\widehat V_i =v \right] &= \exp \left( \widehat \phi_{i,i+1} + \widehat \psi_{i,i+1} v \right), \label{eq:charhatU}\\
        \mathbb E\left[ \exp\left( wU_{t_i,t_{i+1}}\right) \mid V_i =v \right] &= \exp \left(  \phi (t_{i+1} - t_i) +  \psi (t_{i+1} - t_i) v \right), \label{eq:charU}
         \end{align}
         where 
         \begin{align}\label{eq:widehatpsi}
         \widehat \phi_{i,i+1}& = \frac{ac}{\sigma_i b}\left( \frac{e^{b(t_{i+1} -t_i)}- 1}{b} -(t_{i+1} -t_i) \right)   \widehat \psi_{i,i+1} \quad \text{and} \quad 
   \widehat \psi_{i,i+1} = \frac{  1 -  \sqrt{1 - 2 w  \sigma_i^2}}{c \sigma_i},
\end{align}
with $\sigma_i$ given by \eqref{eq:alphasigma}   and      $(\phi,\psi)$ are the functions explicitely given by \eqref{eq:Hestonexplicit}. In particular, $\widehat \psi_{i,i+1}$ is a root of the quadratic polynomial
\begin{align}\label{eq:rootpsi}
    \widehat \psi = w \frac{e^{b(t_{i+1}-t_i)} - 1}{b} + \frac{c\sigma_i}{2} \widehat \psi^2,
\end{align}
and $(\phi,\psi)$ solve the equations 
         \begin{align}
             \phi(t) &= a \int_0^t  \psi(s) ds, \label{eq:Ricvariation1} \\
             \psi(t) &= w\frac{e^{b t} - 1 }{b} +   \frac{c^2}2\int_0^t e^{b(t-s)} \psi^2(s) ds, \quad t\geq 0. \label{eq:Ricvariation}
         \end{align}
\end{proposition}

\begin{proof}
We first prove \eqref{eq:charhatU} using the characteristic function of the Inverse Gaussian distribution recalled in \eqref{eq:IGchar}.  Conditional on $\widehat V_i$, it follows from \eqref{eq:hatUsample} that $\widehat U_{i,i+1}$  follows an Inverse Gaussian distribution with mean parameter $\mu_i = \alpha_i$  and shape parameter $\lambda_i = \frac{\alpha_i^2}{\sigma^2}$. Then, \eqref{eq:alphasigma} gives     $$ \frac{\lambda_i}{\mu_i} = \frac{\alpha_i}{\sigma_i^2} = \widehat V_i \frac{1}{c \sigma_i} + \frac{a}{\sigma_i^2 b} \left( \frac{e^{b(t_{i+1}-t_i)} - 1}b  - (t_{i+1}-t_i) \right)  \quad \text{and} \quad \frac{\mu_i^2}{\lambda_i} = \sigma_i^2, $$
    which plugged in \eqref{eq:IGchar}  yields \eqref{eq:charhatU}. In addition, it is straightforward to check that $\widehat \psi_{i,i+1}$  solves \eqref{eq:rootpsi}. As for \eqref{eq:charU}, it follows from the conditional characteristic function \eqref{eq:hestonchar} applied with $v=0$ between $t_i$ and $t_{i+1}$. In particular, a variation of constant formula on the Riccati equation for $\psi$ in \eqref{eq:Ric2}, with $v=0$, yields \eqref{eq:Ricvariation}.  
\end{proof}

With the help of Proposition~\ref{P:charcomparison}, the choice of the Inverse Gaussian distribution can now be justified by a simple discretization of the Riccati equation \eqref{eq:Ricvariation} as follows. 
\begin{remark}\label{R:charcomparison}
    The first step is to  discretize \eqref{eq:Ricvariation} using the right endpoint rule by  writing $\Delta_i =t_{i+1} -t_i $:
$$ \psi(\Delta_i) \approx  w \frac{e^{b\Delta_i}-1}{b} + \frac{c^2}{2} \int_0^{\Delta_i} e^{b({\Delta_i}-s)}  ds \psi^2(\Delta_i) = w \frac{e^{b\Delta_i}-1}{b} + \frac{c\sigma_i}{2} \psi^2(\Delta_i).   $$
This yields the quadratic equation \eqref{eq:rootpsi} as approximation for $\psi(\Delta_i)$. The root with a non-positive real part is given precisely by $\widehat \psi_{i,i+1}$ in \eqref{eq:widehatpsi}. As for $\widehat \phi_{i,i+1}$ in \eqref{eq:widehatpsi}, for small $\Delta_i$ we have that 
$$ \widehat \phi_{i,i+1} \approx \frac{a\Delta_i}{2} \widehat \psi_{i,i+1},$$
which corresponds to a trapezoidal discretization of  $\phi(\Delta_i)$ in \eqref{eq:Ricvariation1} with $\psi(\Delta_i)$ being approximated by $\widehat \psi_{i,i+1}$.  In other words, the discretization between $t_i$ and $t_{i+1}$ of the Riccati equations   \eqref{eq:Ricvariation1}-\eqref{eq:Ricvariation}  that govern the conditional distribution of the integrated process $U_{t_i,t_{i+1}}$ in \eqref{eq:charU} naturally leads to an Inverse Gaussian distribution of the form \eqref{eq:charhatU} and \eqref{eq:widehatpsi}. \qedsymbol
\end{remark}

Another advantage of our scheme using the Inverse Gaussian distribution is that such distribution emerges as the limiting distribution of the integrated process $U_{0,T}$ in two configurations:

\begin{remark}\label{R:largemeanrev}
\begin{itemize}
    \item In market regimes characterized by high mean reversion and  volatility of volatility:  the results of \cite{mechkov2015fast}  and later extended by \cite{abijaber2024reconciling,mccrickerd2019foundations} have shown that under the parameterization $c = -b\beta $ and $ a = -b\gamma$, with 
$\beta, \gamma > 0 $, as $ b \to \infty$, the distribution of $U_{0,T}$ converges to an Inverse Gaussian distribution. Furthermore, weak convergence holds for the entire process $(U_{0,t})_{t \geq 0}$ towards an Inverse Gaussian Lévy process in  Skorokhod spaces.  This observation suggests that in such market regimes, which are typically encountered when calibrating the Heston model to  the short end of the SPX volatility surface  (see Figure~\ref{fig:calib2013}), the {\normalfont iVi} scheme is expected to perform exceptionally well. In particular, even with a minimal number of time steps—potentially as few as one single step—our approach maintains its accuracy,  as already demonstrated in Figure~\ref{fig:intro}.
\item For large maturities, independently of the parameters, the limiting distribution of $U_{0,T}$ as $T\to \infty$ is also an Inverse Gaussian distribution as shown by \cite*{forde2011large}. Interestingly,  such limiting behavior has been used by \cite{tse2013low}  to justify the choice of an Inverse Gaussian distribution for the  distribution of $U_{t_i,t_{i+1}}$ conditional on the endpoints $(V_{t_i},V_{t_{i+1}})$ in Step 2 of the traditional approach recalled in the introduction.\qed
\end{itemize}
\end{remark}

\section{Numerical illustrations for the integrated process $U$}\label{S:U}

In this section, we examine the performance of the iVi scheme given by Algorithm~\ref{alg:simulation} for quantities involving the accumulated integrated process
$$U_{0,T} = \int_0^T V_s \, ds = \sum_{i=0}^{n-1} \int_{t_i}^{t_{i+1}} V_s \, ds = \sum_{i=0}^{n-1} U_{t_i,t_{i+1}},$$
which can be naturally approximated using our scheme by
\begin{tcolorbox}[colback=gray!20, colframe=gray!80, sharp corners]
$$\widehat{U}_{0,T} := \sum_{i=0}^{n-1} \widehat{U}_{t_i,t_{i+1}}.$$
\end{tcolorbox}
We focus on the following three distributional quantities which  also have a financial meaning:
\begin{enumerate}
    \item \textbf{First moment:} $\mathbb{E}[U_1]$, which corresponds to the variance swap at maturity $T=1$ if $U$ represents the integrated variance. The reference value is computed using:
     \[
    \mathbb{E}[U_T] = V_0 \frac{e^{bT} - 1}{b} + \frac{a}{b} \left( \frac{e^{bT} - 1}{b} - T \right).
    \]
    \item \textbf{Half-moment:} $\mathbb{E}[\sqrt{U_1}]$, representing the volatility swap. For the reference value, we use the inversion formula of the  Laplace transform in \cite{schurger2002laplace}:
    \[
    \mathbb{E}\left[\sqrt{U_1}\right] = \frac{1}{2\sqrt{\pi}} \int_0^\infty \frac{1 - \mathbb{E}[e^{-u U_1}]}{u^{3/2}} \, du,
    \]
    where the Laplace transform is recalled in \eqref{eq:hestonchar}.
    \item \textbf{Laplace transform:} $\mathbb{E}[e^{-U_1}]$, which can also be interpreted as the price of a zero-coupon bond if $V$ models the short rate. The explicit reference value  is recalled in \eqref{eq:hestonchar}.
\end{enumerate}

The parameters for our tests are listed in Table~\ref{tab:parameter_cases}. All three cases correspond to realistic parameter values. Cases 1 and 2 are derived from calibrating the Heston model to SPX market data on two distinct dates, October 10, 2017, and July 3, 2013. The fitted volatility surfaces are   illustrated  on Figures~\ref{fig:calib2017} and \ref{fig:calib2013}. Case 1 places an emphasis on short-dated options. Case 3, on the other hand, corresponds to the first case considered by \cite{andersen2007efficient} and represents the market dynamics for long-dated FX options. The parameter $\rho$ in Table~\ref{tab:parameter_cases} is not used in this section but will become relevant in the next section, where it is applied to the Heston model.

\begin{table}[h!]
\centering
\[
\begin{array}{|c|c|c|c|c|c|}
\hline
\textbf{Case} & V_0 & a & b & c & \rho \\ 
\hline
\text{\textbf{Case 1}} & 0.006 & 17.25\times 0.018 & -17.25 & 2.95 & -0.68 \\ 
\hline
\text{\textbf{Case 2}} & 0.023 & 2.15\times 0.057 & -2.15 & 0.86 & -0.70 \\ 
\hline
\text{\textbf{Case 3}} & 0.04 & 0.5\times 0.04 & -0.5 & 1.0 & -0.9 \\ 
\hline
    \end{array}
\]
\caption{Parameter values for the three cases.}
\label{tab:parameter_cases}
\end{table}

Case 1 is particularly challenging for simulation due to its high volatility of volatility and strong mean-reversion dynamics. Case 3 features relatively low mean-reversion compared to high volatility. Case 2 lies between the two extremes. In all three cases, the Feller condition  is violated, i.e.~$ a - \frac{c^2}{2} < 0$.

For comparison, we also implement the Quadratic-Exponential (QE) scheme of \cite{andersen2007efficient} and the second order scheme (AL) of \cite{alfonsi2010high}..

\begin{figure}[h!]
    \centering
    \includegraphics[width=.9\textwidth]{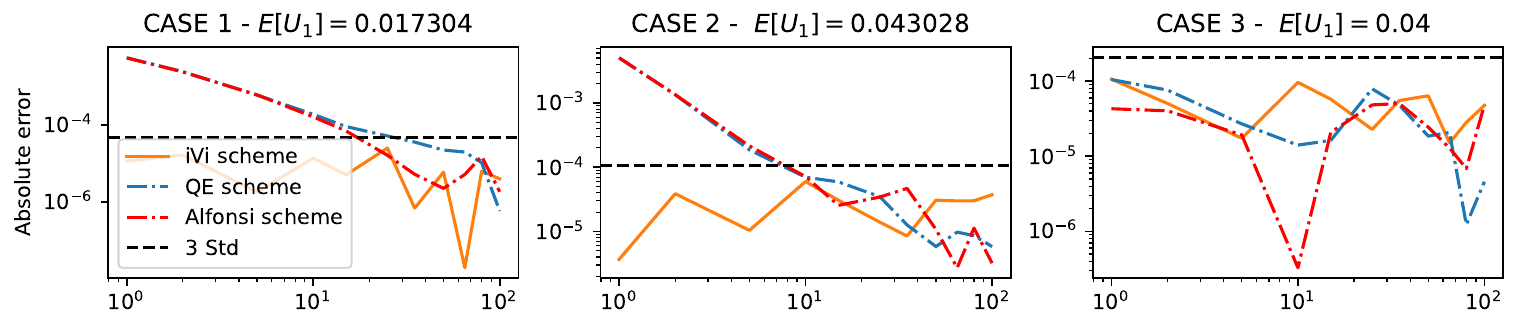} \\
    \includegraphics[width=.9\textwidth]{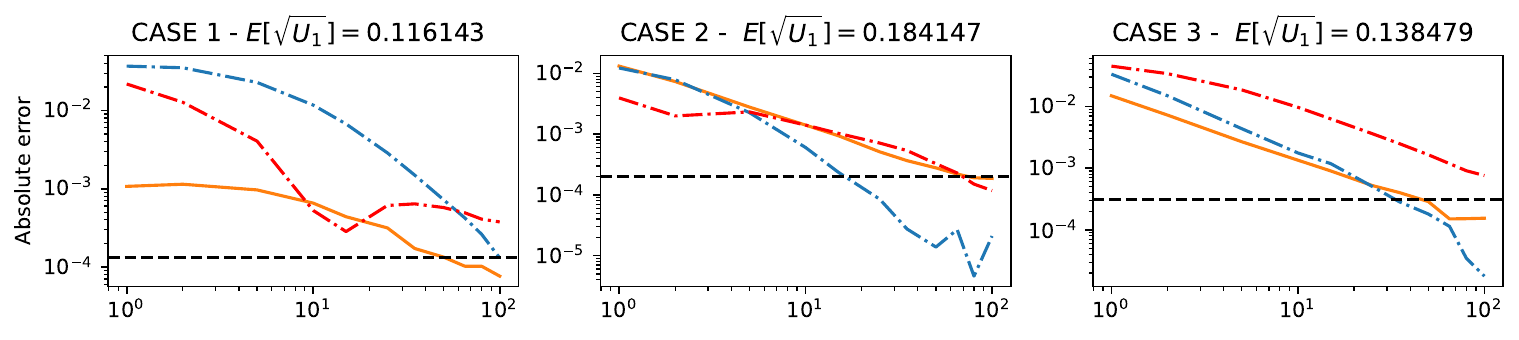} \\
    \includegraphics[width=.9\textwidth]{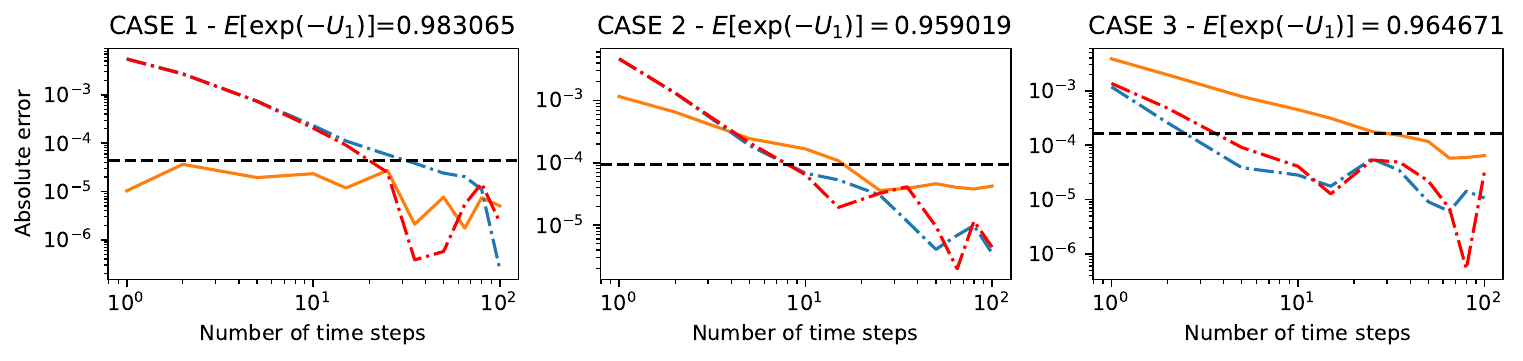}
    \caption{Errors on variance swaps, volatility swaps and Laplace transform of $\hat U_T$ in terms of the number of time steps  for  the three cases with  $T = 1$ and 2 million sample paths.}
    \label{fig:U} 
\end{figure}

Figure~\ref{fig:U} shows the absolute error between the schemes and reference values varying number of times steps on a uniform grid from 1 to 100 in log-log scale. Simulations use 2 million sample paths. The horizontal black dotted line represents three standard deviations of the Monte Carlo estimator, below which comparisons can be considered non-significant. We can observe that:
\begin{itemize}
    \item \textbf{Convergence:} All schemes converge as the number of time steps increases.
    \item \textbf{For \textbf{$\mathbb{E}[U_1]$:}} the iVi scheme achieves accuracy within three standard deviations across all cases, even with a single time step, in line with Proposition~\ref{P:moments}, outperforming the QE and AL schemes for the first two cases.
    \item \textbf{For Case 1:} the iVi scheme delivers highly accurate results for all quantities with a very small number of time steps, including $\mathbb{E}[e^{-U_1}]$, even with one single time step.
    \item \textbf{For Case 2:} The performance of the iVi scheme for the volatility swap and Laplace transform is comparable to other  schemes, the iVi scheme is more accurate when we have less than 5 time steps in 4 out of the 6 cases.
    \item \textbf{For Case 3:} The QE and AL schemes converge faster for the Laplace transform.
\end{itemize}

To sum up, for quantities on the integrated process $U$, 
the iVi scheme converges and competes favorably with the QE and AL schemes,  particularly  in challenging regimes with high volatility-of-volatility and strong mean reversion.

\section{Numerical illustrations  for the Heston model}\label{S:HestonNumerics}
In this section, we test our iVi scheme on   the \cite{heston1993closed} model where the stock price $S$ is given by 
\begin{align}\label{eq:HestonS}
    dS_t = S_t \sqrt{V_t} \left(\rho dW_t + \sqrt{1-\rho^2}dW_t^{\perp}\right), \quad S_0 >0,
\end{align}
where $V$ is given by \eqref{eq:cir}, $\rho \in [-1,1]$ and $W^{\perp}$ is a standard Brownian motion independent of $W$. 

In order to simulate $S$ it suffices to observe that 
$$ \log S_{t_{i+1}} = \log S_{t_i}  - \frac 1 2 U_{t_{i},t_{i+1}} + \rho Z_{t_i, t_{i+1}}  + \sqrt{1-\rho^2} \int_{t_i}^{t_{i+1}} \sqrt{V_s}dW_s^{\perp}, $$
and that conditional on $U_{i,i+1}$, $\int_{t_i}^{t_{i+1}} \sqrt{V_s}dW_s^{\perp} \sim \mathcal N(0, U_{t_i, t_{i+1}})$, for $i=0,\ldots, n-1$.

\begin{tcolorbox}[colback=gray!20, colframe=gray!80, sharp corners]
We can therefore simulate $(\log \widehat S_i)_{i=0,\ldots,n}$ 
using the outputs $(\widehat U, \widehat Z) $ of Algorithm \eqref{alg:simulation} using:
\begin{align}
\log \widehat S_0 &= \log S_0, \\
    \log \widehat S_{i+1} &= \log \widehat S_{i}  - \frac 1 2 \widehat U_{{i},{i+1}} + \rho \widehat Z_{i, {i+1}}  + \sqrt{1-\rho^2} \sqrt{ \widehat U_{{i},{i+1}}} N_i, \quad i=0,\ldots, n-1,  \label{eq:Sii}
\end{align}
where $(N_i)_{i=0,\ldots, n-1}$ are i.i.d.~standard Gaussian random variables. 
\end{tcolorbox}

Clearly, the update rule \eqref{eq:Sii} at the $i$-th step $i$ can be incorporated in Algorithm~\ref{alg:simulation} right after  \eqref{eq:Vii}.   Also, the $(N_i)_{i=0,\ldots, n-1}$ need to be taken independent  from the Gaussian and Uniform random variables used for the sampling of the Inverse Gaussian distribution in Algorithm~\ref{alg:IG_sampling}.

For our numerical experiment, we will consider call options on $S$ for  the three cases of Table~\ref{tab:parameter_cases} with maturity $T=1$ for the first cases and the longer maturity $T=10$ for the third case. Reference values are computed using Fourier inversion techniques on the characteristic function of the log-price which is  known explicitly in the Heston model, see \eqref{eq:hestonchar}.  Simulations use 2 million sample paths and we also benchmark against the QE and AL schemes.  

Figure~\ref{fig:call} displays the absolute error between the schemes and reference values varying number of times steps from 1 to 100 in log-log scale for In-the-Money (ITM), At-the-Money (ATM) and Out-of-The-Money (OTM).  We can  again  observe   the convergence of the iVi scheme.  For ITM call option, the iVi scheme  yields a fast convergence (compared to QE and AL schemes) in all three cases. Again, for Case 1, our scheme   achieves accuracy within three standard deviations with a single time step, in line with Figure~\ref{fig:U}. For ATM and OTM call options, as the strike increases,  the convergence speed seems to deteriorate for our scheme (compared to the QE scheme).  Still for case 1 (first column), our scheme achieves accurate results with one single time step.  
\begin{figure}[h!]
    \centering
    \includegraphics[width=.9\textwidth]{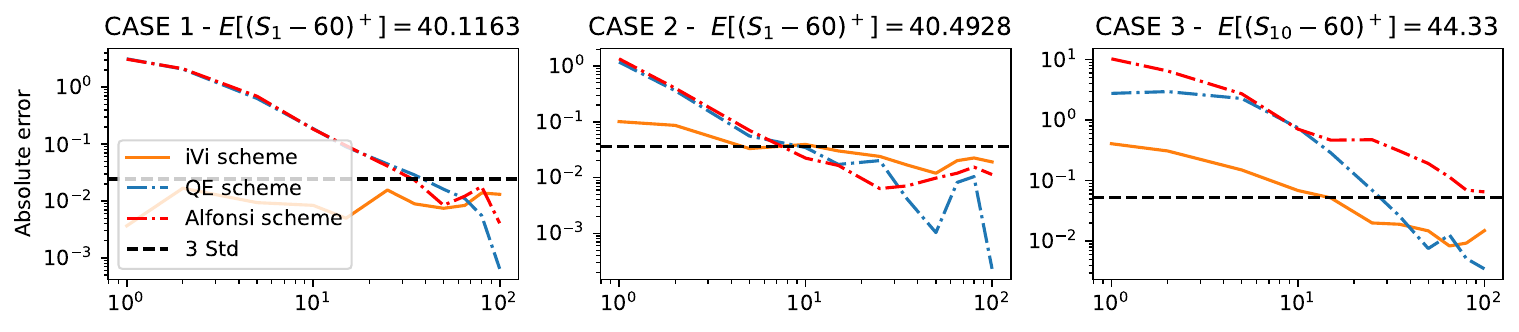} \\
    \includegraphics[width=.9\textwidth]{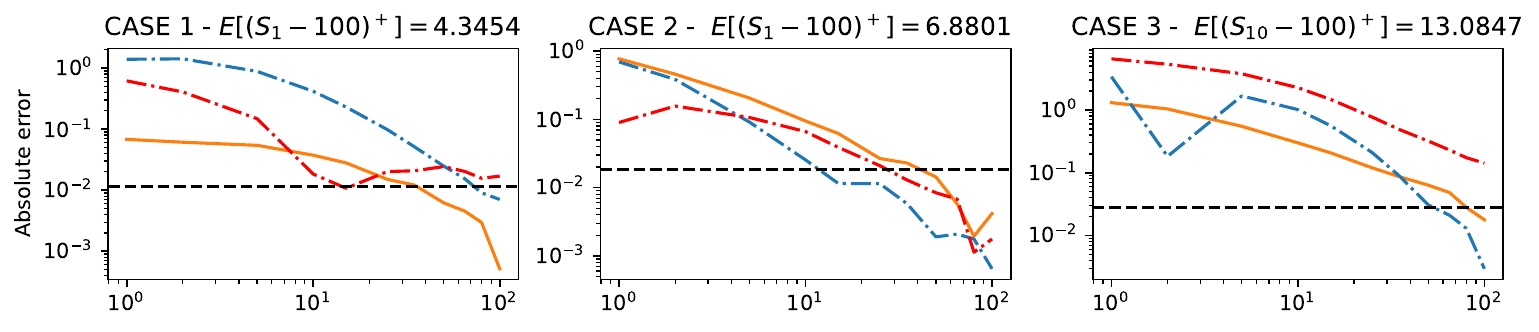} \\
    \includegraphics[width=.9\textwidth]{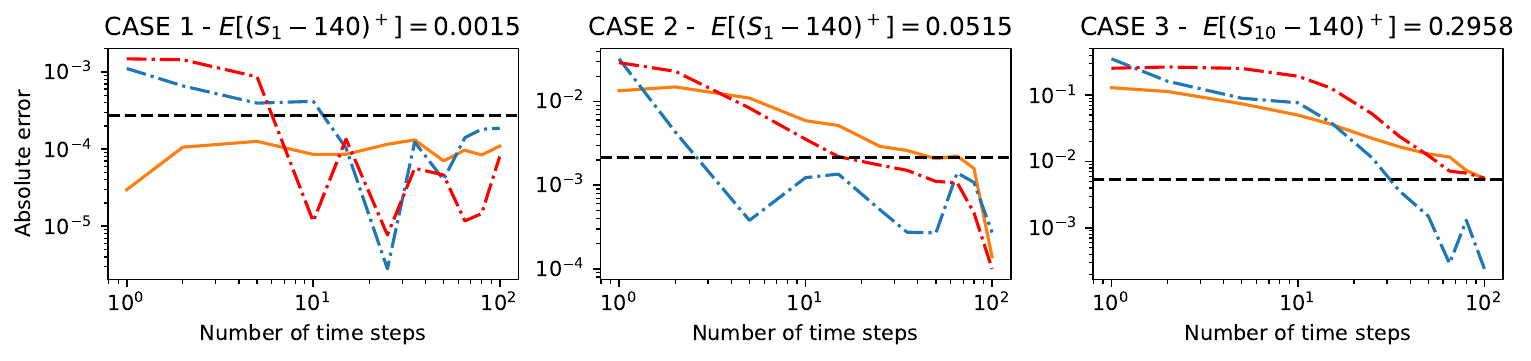}
    \caption{ITM, ATM and OTM call options  on $S$:  error in prices in  terms of number of time steps for the three cases  with 2 million sample paths.}
    \label{fig:call} 
\end{figure}

\begin{figure}[h!]
    \centering
    \includegraphics[width=.3\textwidth]{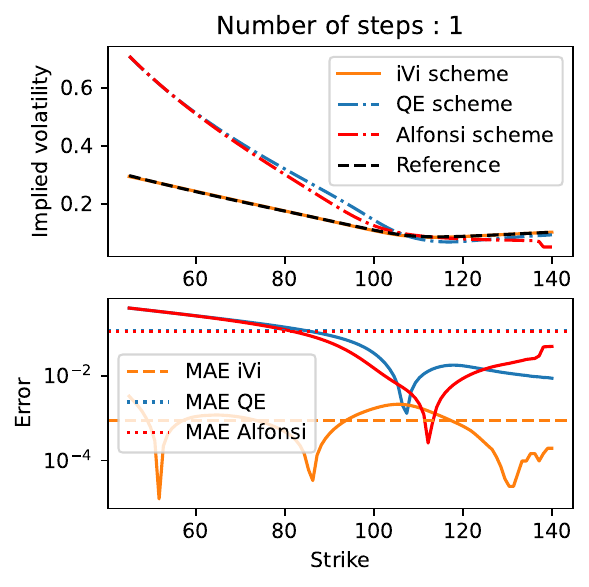} 
    \includegraphics[width=.3\textwidth]{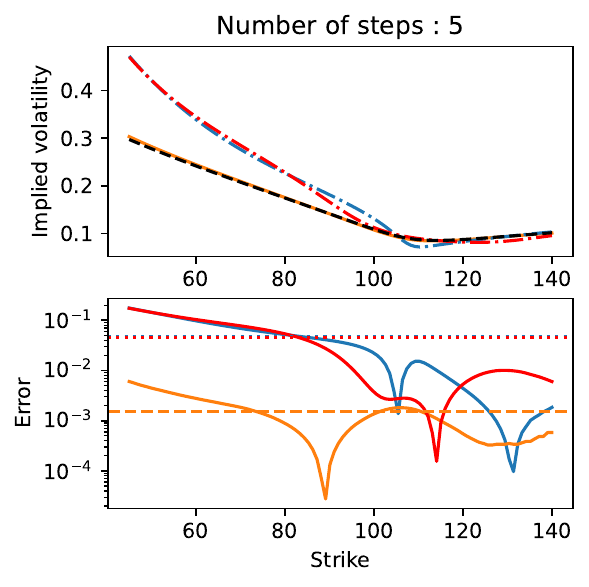} 
    \includegraphics[width=.3\textwidth]{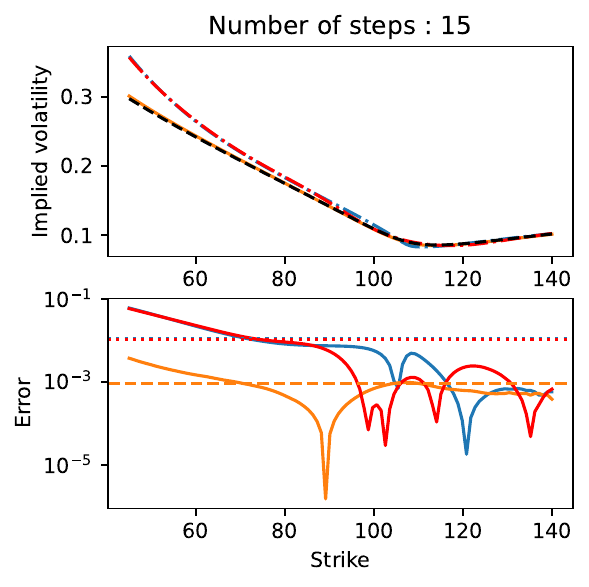}
    \caption{\textbf{Case 1:} Implied volatility slice for $T = 1$ and 2 million sample paths.}
    \label{fig:volcase1} 
\end{figure}

\begin{figure}[h!]
    \centering
    \includegraphics[width=.3\textwidth]{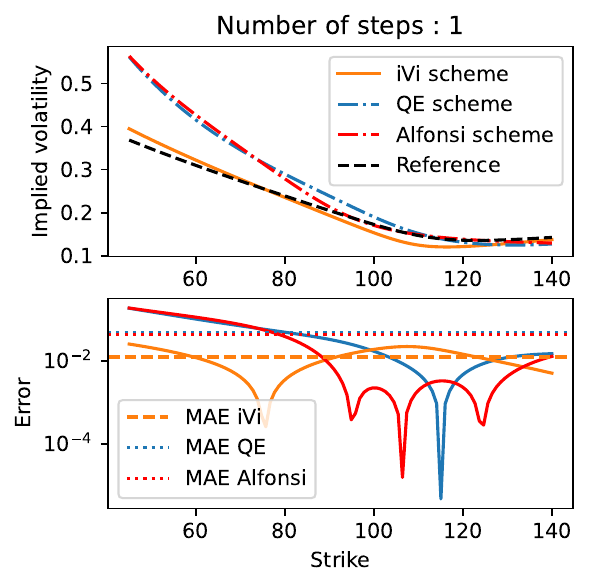} 
    \includegraphics[width=.3\textwidth]{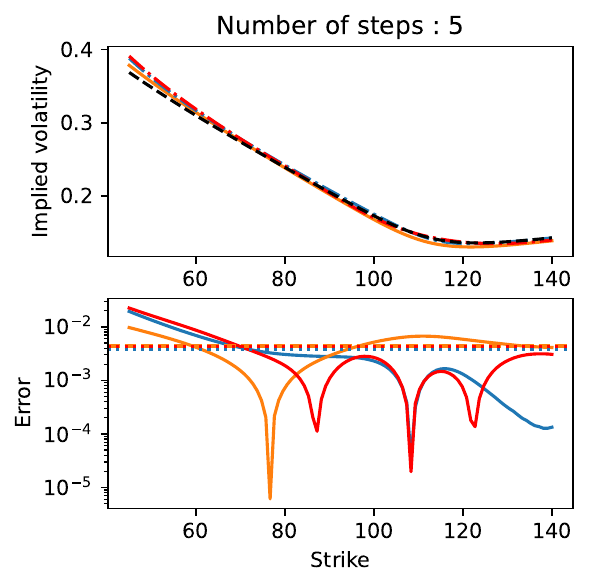} 
    \includegraphics[width=.3\textwidth]{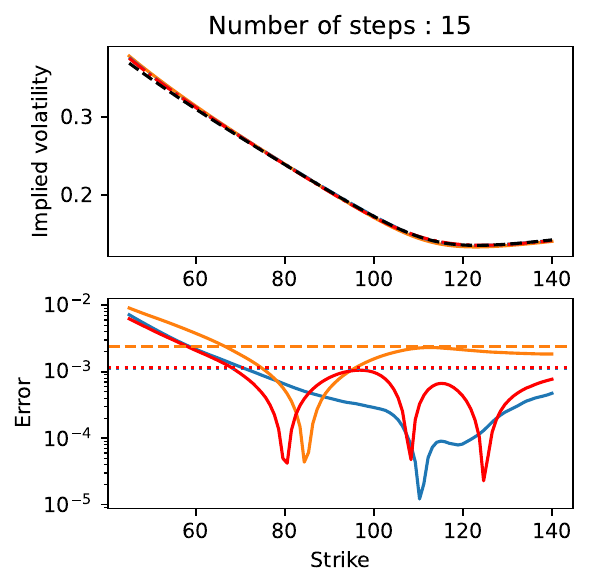}
    \caption{\textbf{Case 2:} Implied volatility slice for   $T = 1$ and 2 million sample paths.}
    \label{fig:volcase2} 
\end{figure}

\begin{figure}[h!]
    \centering
    \includegraphics[width=.3\textwidth]{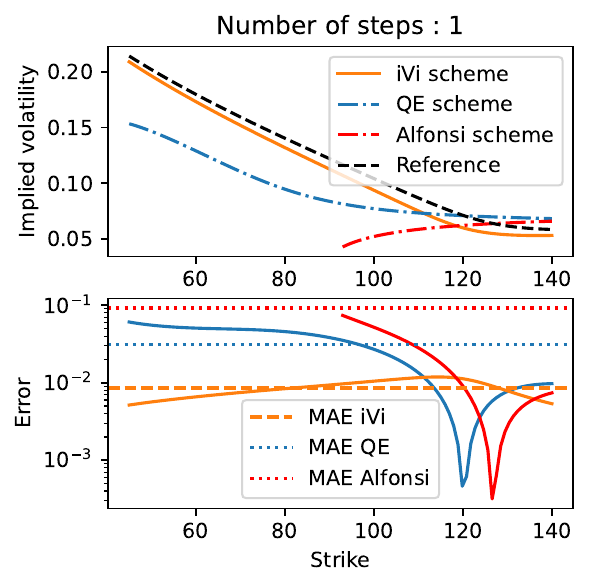} 
    \includegraphics[width=.3\textwidth]{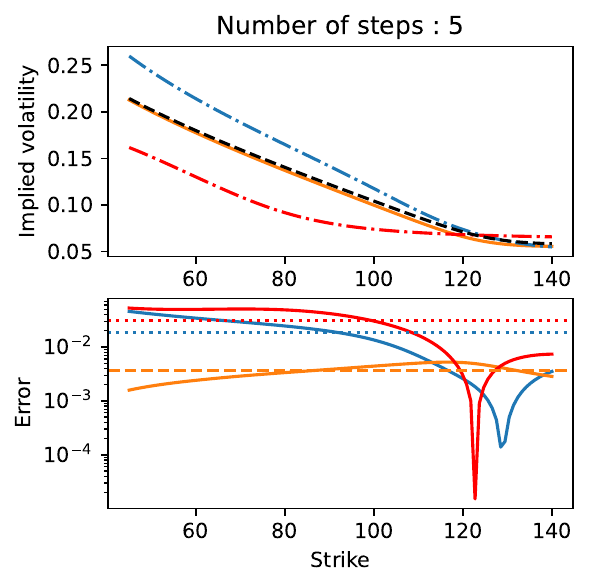} 
    \includegraphics[width=.3\textwidth]{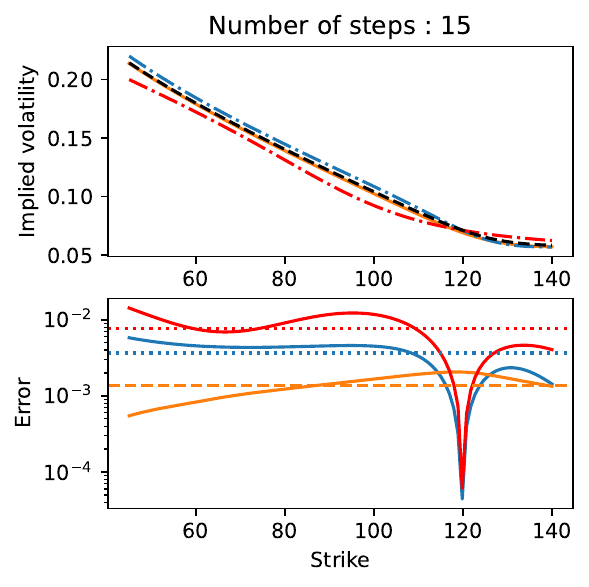}
    \caption{\textbf{Case 3:} Implied volatility slice for   $T = 10 $ and 2 million sample paths.}
    \label{fig:volcase3} 
\end{figure}

Figures~\ref{fig:volcase1}, \ref{fig:volcase2}, and \ref{fig:volcase3} display the full implied volatility slices for the three cases, corresponding to simulations with 1, 5, and 15 time steps. These figures also include the absolute errors in implied volatility and the mean absolute errors across the entire slice. The iVi scheme demonstrates excellent accuracy with just 1 and 5 time steps, particularly for Case 1. With 15 time steps, the slices of the iVi scheme and the reference values become nearly indistinguishable for all three cases. Compared to the QE and AL schemes, the iVi scheme is always more accurate for (deep) ITM call options, and for the whole volatility slice  our scheme outperforms, in terms of Mean Absolute Error (MAE), for 7 out of the 9 plots.  

Additionally, for Case 1, Figure~\ref{fig:intro} provides  six slices of the volatility surface, computed using the iVi scheme with only a single time step per slice, highlighting the scheme's efficiency in approximating the entire surface with minimal discretization in high mean-reversion and volatility-of-volatility market regimes.

Finally, to illustrate the relevance of the iVi scheme for practical applications, for instance, when the Heston  model is used  as a component of a Local Stochastic Volatility model, a small number of paths is typically used with a fixed number of time steps. Figure \ref{fig:call2} shows the option prices obtained from the three schemes using fewer paths (ranging from 10k to 200k) with 50 time steps. Among the three schemes, the iVi scheme appears to produce the most stable results across different numbers of sample paths.

\begin{figure}[h!]
    \centering
    \includegraphics[width=.9\textwidth]{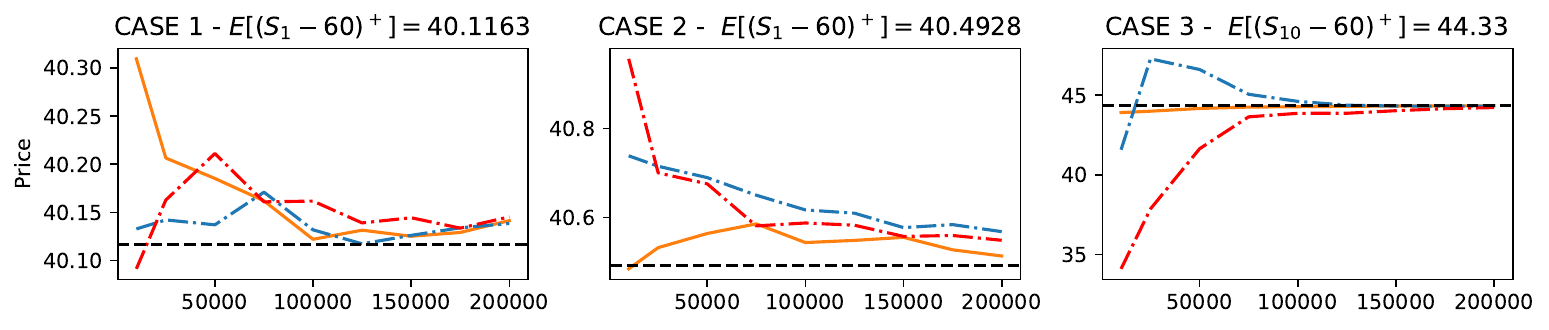} \\
    \includegraphics[width=.9\textwidth]{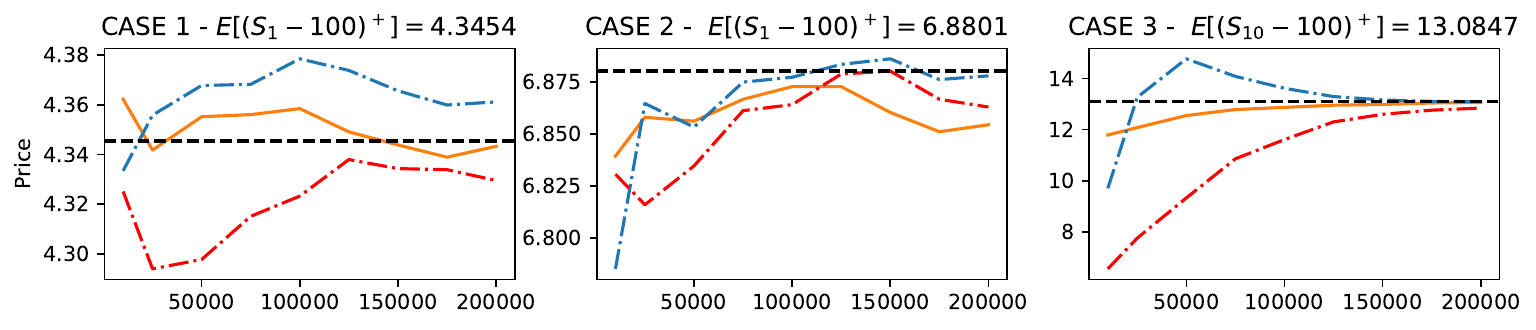} \\
    \includegraphics[width=.9\textwidth]{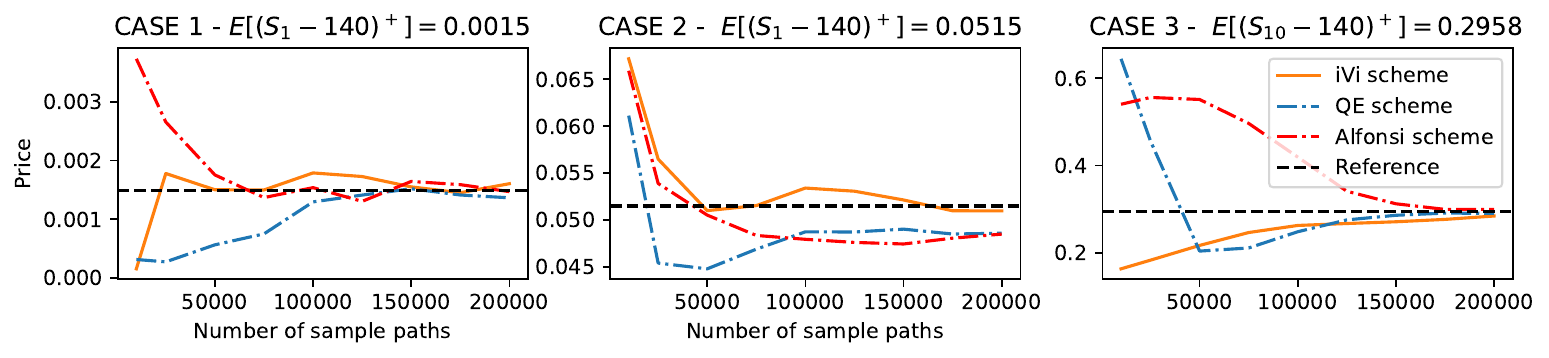}
    \caption{ITM, ATM and OTM call options  on $S$:  prices in  terms of the  number of sample paths for the three cases  with 50 time steps.}
    \label{fig:call2} 
\end{figure}

All in all,  the iVi scheme seems pretty competitive across several parameter sets, strikes and maturities when pricing call options in the Heston model, especially with a little number of time steps. It is worth noting that second and third-order schemes of \cite{alfonsi2010high}, as well as the QE scheme, all rely on switching mechanisms, which can be less robust, especially in the high volatility-of-volatility and fast mean-reversion regimes. In contrast, the iVi scheme already incorporates the Inverse Gaussian limiting distribution in such regimes, and appears to perform more reliably, without the need of switching mechanisms.

\section{Extensions}
Extension of the method to multivariate settings, even in affine models such as Wishart-type processes, is not straightforward. The Inverse Gaussian distribution used here is inherently one-dimensional, and extending it to higher dimensions while preserving tractability appears nontrivial. However, the scheme does naturally extend to path-dependent settings and can be used to efficiently simulate classes of Volterra and rough Heston models. This has been developed very recently in \cite{abijaber2025simulating}. In particular, \cite[Theorem 2.2 and Corollary 2.3]{abijaber2025simulating} guarantee the theoretical convergence of the iVi scheme.

\appendix
\section{Inverse Gaussian and its sampling}\label{A:IG}

The Inverse Gaussian distribution, also known as the Wald distribution, is a continuous probability distribution on $\R_+$ with two parameters:    $ \mu > 0$ (mean parameter) and $\lambda > 0$ (shape parameter). The probability density function  of the Inverse Gaussian distribution is given by:
\[
f(x ; \mu, \lambda) = \sqrt{\frac{\lambda}{2\pi x^3}} \exp \left( -\frac{\lambda (x - \mu)^2}{2\mu^2 x} \right), \quad x > 0.
\]
We  denote  $X\sim IG(\mu, \lambda)$ to indicate that $X$ is a random variable with an inverse Gaussian distribution with mean parameter $\mu$ and shape parameter $\lambda$. Its characteristic function is given by 
\begin{align}\label{eq:IGchar}
\mathbb E \left [ \exp\left( wX\right) \right]= \exp\left(\frac{\lambda}{\mu} \left(1 - \sqrt{1 - \frac{2w\mu^2 }{\lambda}}\right)\right),
\end{align}
for all $w\in \mathbb C$ such that  $\Re(w)\leq 0$.  In addition, its mean is given by 
\begin{align}\label{eq:IGmean}
    \mathbb E[X] = \mu.
\end{align}

Inverse Gaussian random variables can be simulated easily using one Gaussian random variable and one Uniform random variable using  an acceptance-rejection step as shown in \citet*{michael1976generating}, we recall the algorithm here.

\begin{algorithm}[H]
\caption{Sampling from the Inverse Gaussian Distribution}\label{alg:IG_sampling}
\begin{algorithmic}[1]
\State \textbf{Input:} Parameters \( \mu > 0 \), \( \lambda > 0 \).
\State \textbf{Output:} Sample \( IG \) from the Inverse Gaussian distribution.

\State   
Generate \( \xi \sim \mathcal{N}(0, 1) \) and compute \( Y = \xi^2 \).

\State Compute the candidate value \( X \):  
$$
X = \mu + \frac{\mu^2 Y}{2\lambda} - \frac{\mu}{2\lambda} \sqrt{4\mu\lambda Y + \mu^2 Y^2}.
$$

\State Generate a uniform random variable:  
Sample \( \eta \sim \text{Uniform}(0, 1) \).

\State Select the output:  
\If{$ \eta \leq \frac{\mu}{\mu + X} $}
    \State Set the output \( IG = X \).
\Else
    \State Set \( IG = \frac{\mu^2}{X} \).
\EndIf
\end{algorithmic}
\end{algorithm}

\section{Heston's characteristic function}\label{S:Heston}

We recall the expression of the joint characteristic function of the integrated variance $U_{t,T}=\int_t^T V_s ds$ and $\log S$ in the \cite{heston1993closed} model in \eqref{eq:HestonS} with $V$ given by \eqref{eq:cir}.  Let $u,w \in  \mathbb{C}$ be such that
\begin{equation}
    \Re w + \frac{1}{2} \left((\Re u)^2 -\Re u\right)  \le 0.
\end{equation}
The joint conditional characteristic function of $(\log S, U)$ is given by 
\begin{equation} \label{eq:hestonchar}
 \mathbb{E} \left[ \left. \exp\left(  u \log \frac{S_T}{ S_t} + w U_{t,T}\right) \right| V_t \right] = \exp\left(\phi \left( T-t \right) + \psi \left( T-t \right) V_t \right), \quad t \leq T,
\end{equation}
where $(\phi, \psi)$  are explicitly given by
\begin{equation}\label{eq:Hestonexplicit}
\begin{aligned}
\psi(t)&= \frac{\beta(u)-D(u,w)}{c^2}\frac{1-e^{-D(u,w)t}}{1-G(u,w)e^{-D(u,w)t}}, \\
\phi(t)&= \frac{a}{c^2}\left((\beta(u)-D(u,w))t - 2 \log\left( \frac{G(u,w)e^{-D(u,w)t}-1}{G(u,w)-1} \right)\right), \quad t\geq 0,  \\
\beta(u) &= - b  - u \rho c, \quad  D(u,w)= \sqrt{\beta(u)^2  + c^2 ( -2w + u - u^2)}, \quad G(u,w) =\frac{\beta(u)-D(u,w)}{\beta(u)+ D(u,w)}.
\end{aligned}    
\end{equation}

In particular,   $(\phi,\psi)$ solve  the following system of  Riccati equations 
\begin{align}
\phi'(t) & = a \psi(t), \quad \phi(0) = 0, 
 \label{eq:Ric1}\\
   \psi'(t) & =   \frac{c^2}{2}  \psi^2(t) +  \left( \rho c u  + b \right) \psi(t)  + w + \frac{u^2-u}{2}, \qquad  \psi(0) = 0. \label{eq:Ric2}
\end{align}

See for example \cite[Theorem 2.1]{abijaber2024reconciling} and \cite[Chapter 2]{gatheral2011volatility}, \cite{albrecher2007little} for the explicit derivation of the formulas, and \cite{heston1993closed} for the initial derivation.

\section{Calibrated volatility surfaces on SPX data}\label{S:calibrated}

\begin{figure}[H]
    \centering
    \includegraphics[width=0.66\textwidth]{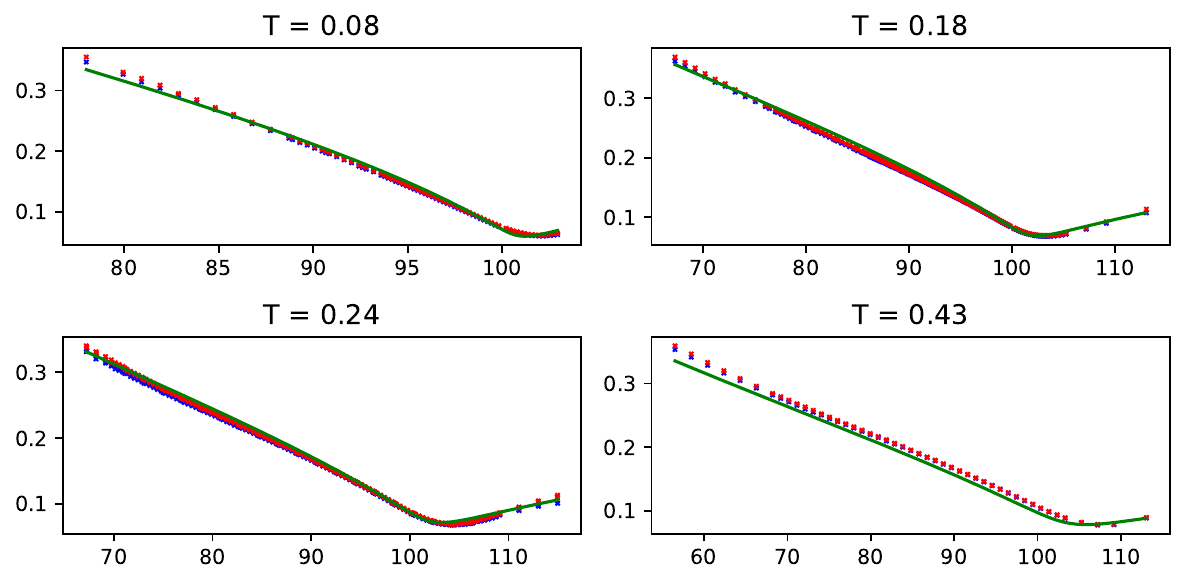} \\
    \caption{Implied volatility surface of the calibrated Heston model (green) on the market bid and ask implied volatilities (red and blue dots) on  October 10, 2017. Calibrated parameters correspond to Case 1 in Table~\ref{tab:parameter_cases}.  Market
data from the CBOE website \url{https://datashop.cboe.com/}.}
    \label{fig:calib2017} 
\end{figure}

\begin{figure}[H]
    \centering
    \includegraphics[width=.9\textwidth]{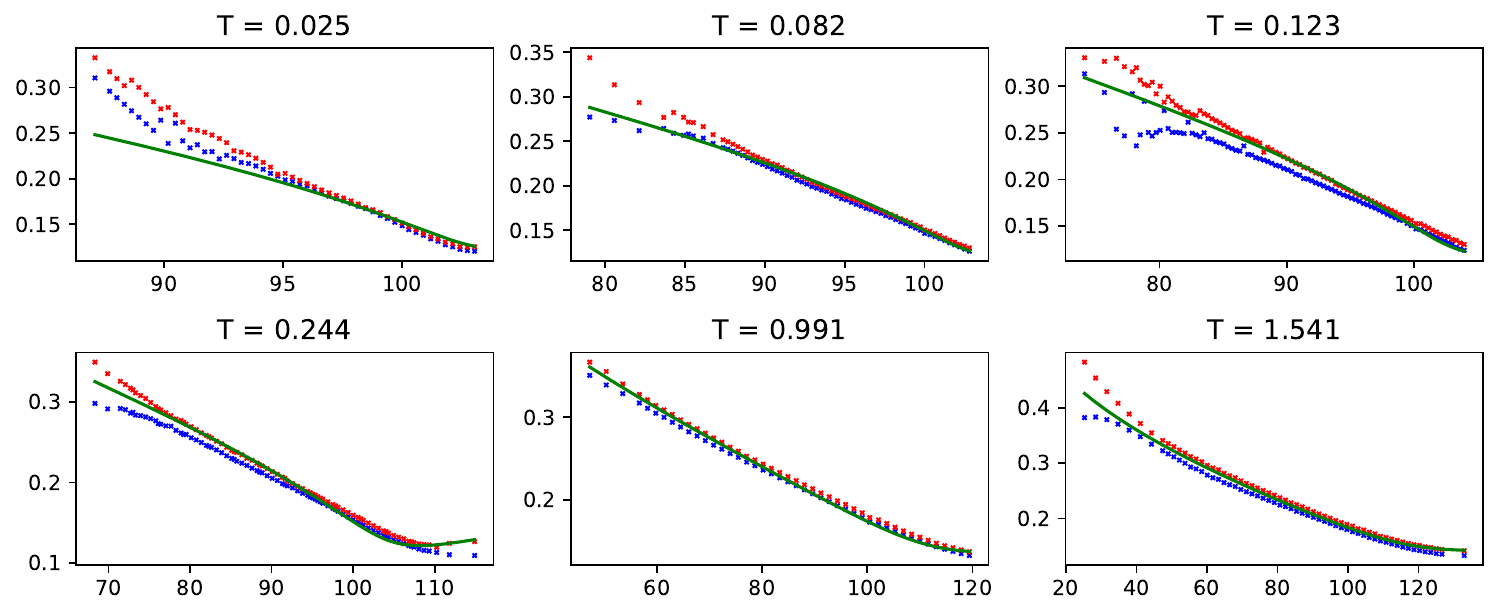} \\
    \caption{Implied volatility surface of the calibrated Heston model (green) on the market bid and ask implied volatilities (red and blue dots) on July 3, 2013. Calibrated parameters correspond to Case 2 in Table~\ref{tab:parameter_cases}. Market
data from the CBOE website \url{https://datashop.cboe.com/}.}
    \label{fig:calib2013} 
\end{figure}

\bibliographystyle{plainnat}
\bibliography{bibl}

\end{document}